\documentclass[journal,onecolumn,10pt]{IEEEtran}
\usepackage{amsmath,amssymb,bm}
\usepackage{graphicx}
\usepackage{cite}
\usepackage{tikz}
\usetikzlibrary{arrows.meta,positioning,shapes.geometric}
\usepackage{microtype}
\usepackage{url}
\newtheorem{theorem}{Theorem}
\newtheorem{proposition}{Proposition}
\newtheorem{corollary}{Corollary}
\newcommand{\CN}{\mathcal{CN}}
\newcommand{\HH}{\mathcal{H}}
\newcommand{\Qb}{\overline F}
\newcommand{\E}{\mathbb E}
\title{{\normalfont\normalsize\bfseries Accepted for publication in IEEE Wireless Communications Letters, Copyright transform to IEEE}\\[1.0ex]
Exact Bistatic Detection With Dual-Use AF Relays in SAGINs}
\author{Changhao He,~\IEEEmembership{Graduate Student Member,~IEEE,} and Xiaojuan Zhang,~\IEEEmembership{Senior Member,~IEEE} %
\thanks{C. He is with King Abdullah University of Science and Technology (KAUST), Thuwal, Saudi Arabia (e-mail: changhao.he@kaust.edu.sa). X. J. Zhang is with the Institute of Advanced Intelligence and Computing, Agency for Science, Technology and Research (A*STAR), Singapore (e-mail: xiaojuanzhang@ieee.org).}
 \thanks{This work was supported by the Agency for Science, Technology and Research (A*STAR), Singapore under the Semiconductor RIE Flagship Programme (Integrated Circuit Design \& Systems), grant No. H26-MSE1535.}}

\begin{document}
\raggedbottom
\maketitle


\begin{abstract}
We study bistatic UAV detection using dual-use amplify-and-forward (AF) relays in satellite--aerial--ground integrated networks. The target reflects the pilot, payload, and relay noise, allowing the exact finite-sample Neyman--Pearson detection probability to be derived in noncentral chi-square form under Gaussian payload signaling. The AF cascade introduces a variance signature even without payload transmission and leads to a rate-constrained pilot allocation, distinct relay-power limits, and phase-dependent variance changes under residual direct-path interference. For a decorrelated or suppressed direct path, we further derive a phase-only generalized likelihood ratio test with exact single-integral error probabilities at an assumed target magnitude; the correlated case retains a known-coefficient detection law. A 15-GHz low-Earth-orbit link budget and Monte Carlo simulations validate the analysis.
\end{abstract}

\begin{IEEEkeywords}
Amplify-and-forward relay, bistatic sensing, integrated sensing and communications, intruder UAV detection, SAGIN.
\end{IEEEkeywords}


\section{Introduction}
Low-altitude networks need to detect unauthorized unmanned aerial vehicles (UAVs), but deploying dedicated radar systems can be costly. Integrated sensing and communications (ISAC) offers an alternative by reusing communication infrastructure and waveforms \cite{liu2022isac,xiong2023tradeoff}. Recent studies have considered networked low-altitude ISAC \cite{cheng2025networked}, hybrid deterministic/Gaussian bistatic waveforms \cite{song2026lae,xie2025bistatic}, and amplification-assisted detection using an active intelligent reflecting surface (IRS) \cite{song2025activeirs}. Amplify-and-forward (AF) UAV relays can also support satellite--aerial--ground integrated networks (SAGINs) when direct satellite-to-ground links are blocked \cite{bhola2026amud}.

The closest model is \cite{song2026lae}, which derives an exact Neyman--Pearson (NP) detector for a direct deterministic/Gaussian waveform without AF gain normalization. Reference \cite{xie2025bistatic} considers a similar waveform but uses a generalized likelihood ratio test (GLRT) for unknown target gain with asymptotic analysis. Here, we retain exact finite-$L$ performance while maximizing over the unknown echo phase at an assumed magnitude. Reference \cite{song2025activeirs} considers target-reflected active-IRS amplifier noise but has neither an AF payload nor a served user. Our contribution is the AF-cascade specialization and its communication--sensing consequences, rather than a new Gaussian detector.


\looseness=-1
Specifically, let $E_\mu$ denote the coherent pilot-echo energy over an $L$-symbol dwell, $\Delta$ the target-induced per-symbol variance increment and $\beta$ the relay gain. The AF mapping is
\begin{align}
 E_\mu&\triangleq L|a|^2\beta^2|g|^2p_d,\qquad
 \Delta\triangleq |a|^2\beta^2(|g|^2p_g+\sigma_r^2),\nonumber\\
 \beta^2&\triangleq\frac{P_r}{|g|^2P_s+\sigma_r^2},
\label{eq:afmap}
\end{align}
where $a$ is the target coefficient, $g$ the satellite--relay channel, $p_d$ and $p_g$ the pilot and payload powers, $P_s$ and $P_r$ the satellite and relay transmit powers, and $\sigma_r^2$ the relay-noise variance. Unlike direct transmission, the relay noise is also reflected by the target, giving $\Delta>0$ even for $p_g=0$. Moreover, $\beta$ couples $P_s$, $P_r$, the user signal-to-interference-plus-noise ratio and the sensing performance. This leads to a rate-constrained pilot split, distinct relay-power limits for matched-filter and full NP detection, and direct-path interference whose correlation with the target echo can reverse the optimal decision tail. We also derive an exact phase-only GLRT at an assumed target magnitude. Fig.~\ref{fig:system} shows the geometry, with illumination provided by a low-Earth-orbit (LEO) satellite through the AF relay; the analysis starts after range--Doppler cell formation, conditional on a coherent-cell target coefficient.

\begin{figure}[t]
\centering













\includegraphics[width=0.70\linewidth,keepaspectratio]{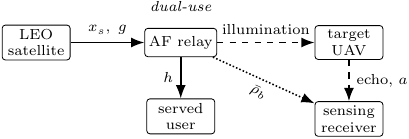}

\caption{Dual-use AF-relay bistatic geometry; solid: communication, dashed: illumination/echo, dotted: residual relay--receiver direct path.}
\label{fig:system}
\end{figure}

\section{AF Cascade and Exact Detector}
We test target absence ($\HH_0$) against target presence
($\HH_1$). Over a coherent dwell of $L$ symbols, the
satellite transmits $x_s[\ell]=\sqrt{p_d}c[\ell]+\sqrt{p_g}s[\ell]$
with $p_d+p_g=P_s$, where $|c[\ell]|=1$ is the known pilot
and $s[\ell]\sim\CN(0,1)$ is the Gaussian payload.
Finite-alphabet signaling makes $\HH_1$ a Gaussian mixture, so the derived law becomes approximate, although $\HH_0$ and its false-alarm calibration remain exact under ideal DPI suppression. At the setting of Fig.~\ref{fig:results}(b), $10^6$-trial simulations show that quadrature phase-shift keying (QPSK) agrees with the Gaussian law within $\pm0.001$; QPSK markers are also shown. All channels and $a$ remain constant over the dwell,
residual delay and Doppler errors are neglected, and payload
and noise samples are independent across $\ell$. The relay
uses a real gain $\beta\ge0$ given by \eqref{eq:afmap} and
forwards $x_r[\ell]=\beta(gx_s[\ell]+n_r[\ell])$ with
$n_r[\ell]\sim\CN(0,\sigma_r^2)$. The served user, with
channel $h$ and noise variance $\sigma_u^2$, treats the pilot
as interference, giving

\begin{equation}
 \gamma_c=\frac{|h\beta g|^2p_g}{|h\beta g|^2p_d+|h\beta|^2\sigma_r^2+\sigma_u^2},
\label{eq:sinr}
\end{equation}
and $R_c=\tau\log_2(1+\gamma_c)$ with resource fraction
$\tau\in(0,1]$; the instantaneous (full-duplex) relaying
in \eqref{eq:raw} gives $\tau=1$.

After subtracting the known direct-path pilot, we first model the residual as decorrelated from the target-cell waveform. In Table~\ref{tab:budget}, the target lies $500$ m off the relay--receiver baseline, giving an excess path of $236$ m compared with the total-path resolution $c/B=2998$ m ($c$ is the speed of light), or $0.08$ cell. The decorrelated residual is therefore an idealized cancellation model; Section~\ref{sec:corr} considers the correlated case. With $\delta_0=0$, $\delta_1=1$,
\begin{align}
 y[\ell]={}&\delta_i a\beta\!\left(g\sqrt{p_d}c[\ell]+g\sqrt{p_g}s[\ell]+n_r[\ell]\right)\nonumber\\
 &+\bar\rho_b\beta\!\left(g\sqrt{p_g}\widetilde s[\ell]+\widetilde n_r[\ell]\right)+n_e[\ell],
\label{eq:raw}
\end{align}
under $\HH_i$, where $n_e\sim\CN(0,\sigma_e^2)$ and the residual direct-path interference (DPI) coefficient $\bar\rho_b\triangleq\rho b$ represents the relay--receiver channel $b$ scaled by the suppression residual $\rho$. The tilded variables represent the residual waveform after range--Doppler projection or cancellation. They are independent of the target-cell waveform but have the same payload-plus-relay-noise variance. Thus, $\HH_0:y[\ell]\sim\CN(0,\sigma_0^2)$ and $\HH_1:y[\ell]\sim\CN(\mu_\ell,\sigma_1^2)$, with $v=|g|^2p_g+\sigma_r^2$,
\begin{align}
 \mu_\ell&=a\beta g\sqrt{p_d}c[\ell],\quad
 \sigma_0^2=\sigma_e^2+|\bar\rho_b\beta|^2v,\nonumber\\
 \sigma_1^2&=\sigma_0^2+\Delta,\qquad \Delta=|a\beta|^2v>0.
\label{eq:moments}
\end{align}
\emph{Fading:} The detection law is conditional on $(a,g,\bar\rho_b)$. If the relay tracks $|g|^2$ over each dwell, as assumed in \eqref{eq:afmap}, feeder-link fading has little effect: $E_\mu\to L|a|^2p_dP_r/P_s$ and $\Delta\to|a|^2P_rp_g/P_s$ for $|g|^2P_s\gg\sigma_r^2$, satisfied by $56.45$ dB in Table~\ref{tab:budget}. Varying $|g|^2$ over $\pm20$ dB changes $P_D$ by less than $9.4\times10^{-5}$. Target-path fluctuation is examined in Fig.~\ref{fig:results}(h).


\emph{Operating regime:} \eqref{eq:radarbudget} gives a
\emph{nominal} $|a|$ for an assumed radar cross section (RCS)
but provides no phase information. The following analysis
therefore serves as a conditional known-coefficient benchmark;
Section~\ref{sec:mismatch} extends it to an unknown phase
using a GLRT and examines magnitude uncertainty for both
detectors.


From \eqref{eq:afmap}, $E_\mu=\sum_\ell|\mu_\ell|^2=L|a\beta g|^2p_d$. Ignoring terms independent of the observations, the log-likelihood ratio is monotone in
\begin{equation}
 T=\frac{\Delta}{\sigma_0^2\sigma_1^2}\sum_\ell|y[\ell]|^2+
 \frac{2}{\sigma_1^2}\sum_\ell\Re\{\mu_\ell^*y[\ell]\},
\label{eq:T}
\end{equation}
which combines received-signal energy and coherent pilot correlation. With $w=\Delta/(\sigma_0^2\sigma_1^2)$ and $\kappa=\sigma_0^2/\Delta$, completing the square gives
\begin{equation}
 T=wU-w\kappa^2E_\mu,\qquad U=\sum_{\ell=1}^{L}|y[\ell]+\kappa\mu_\ell|^2.
\label{eq:square}
\end{equation}
Define $\Qb_{\nu,\lambda}(x)=\Pr\{\chi'^2_\nu(\lambda)>x\}$, with inverse $\Qb^{-1}_{\nu,\lambda}$ in $x$.


\begin{theorem}[Exact finite-$L$ probability]
Under $\HH_i$,
\begin{equation}
 \frac{2U}{\sigma_i^2}\sim\chi'^2_{2L}(\lambda_i),\quad
 \lambda_i=\frac{2\sigma_i^2E_\mu}{\Delta^2},~i\in\{0,1\},
\label{eq:ncx2}
\end{equation}
and for false-alarm probability $P_{\rm FA}=\alpha$,
\begin{equation}
 P_D=\Qb_{2L,\lambda_1}\!\left(\frac{\sigma_0^2}{\sigma_1^2}
 \Qb^{-1}_{2L,\lambda_0}(\alpha)\right).
\label{eq:pd}
\end{equation}
\end{theorem}
\emph{Proof:} For $z_\ell=y[\ell]+\kappa\mu_\ell$, the means are $\E[z_\ell\mid\HH_0]=\kappa\mu_\ell$ and $\E[z_\ell\mid\HH_1]=(1+\kappa)\mu_\ell=(\sigma_1^2/\Delta)\mu_\ell$, with variances $\sigma_0^2$ and $\sigma_1^2$, respectively. Since the shifted observations remain independent proper complex Gaussian, their squared norm follows a noncentral chi-square distribution with $2L$ degrees of freedom and the noncentralities in \eqref{eq:ncx2}. Setting the threshold under $\HH_0$ gives \eqref{eq:pd}. \hfill$\square$


A Gaussian approximation is obtained by taking $U\approx\mathcal N(m_i,v_i)$, where $q_0=\kappa$, $q_1=1+\kappa$,
\begin{equation}
 m_i=L\sigma_i^2+q_i^2E_\mu,~~v_i=L\sigma_i^4+2\sigma_i^2q_i^2E_\mu,
\label{eq:gaussmom}
\end{equation}
giving $P_D^{\rm G}=Q((m_0+\sqrt{v_0}Q^{-1}(\alpha)-m_1)/\sqrt{v_1})$, with $Q(\cdot)$ denoting the Gaussian tail function. A matched-filter (MF) benchmark and, for Fig.~\ref{fig:results}(c), the energy statistic $Z=\sum_\ell|y[\ell]|^2$ have the exact detection probabilities
\begin{align}
 P_D^{\rm MF}&=Q\!\left(\frac{\sigma_0}{\sigma_1}Q^{-1}(\alpha)-\frac{\sqrt{2E_\mu}}{\sigma_1}\right),
\label{eq:mf}\\
 P_D^{\rm E}&=\Qb_{2L,2E_\mu/\sigma_1^2}\!\left(
 \frac{\sigma_0^2}{\sigma_1^2}\Qb^{-1}_{2L,0}(\alpha)\right),
\label{eq:energy}
\end{align}
which characterize the two components combined by the NP test.


\section{Rate, Power, and Direct-Path Consequences}
A minimum rate $R_{\min}$ requires $\gamma_c\ge\gamma_{\min}=2^{R_{\min}/\tau}-1$. With $\mathcal A=|h\beta g|^2$, $\mathcal B=|h\beta|^2\sigma_r^2+\sigma_u^2$ and $p_g=P_s-p_d$, $\gamma_c(p_d)=\mathcal A(P_s-p_d)/(\mathcal Ap_d+\mathcal B)$ decreases strictly with pilot power. Thus, the rate requirement is feasible if and only if
\begin{equation}
 \gamma_{\min}\le\gamma_c(0)=\frac{\mathcal AP_s}{\mathcal B},
\label{eq:ratefeas}
\end{equation}
and the maximum feasible pilot power is
\begin{equation}
 p_d^{\max}=\frac{\mathcal AP_s-\gamma_{\min}\mathcal B}{\mathcal A(1+\gamma_{\min})}.
\label{eq:pdmax}
\end{equation}
\eqref{eq:sinr} treats the pilot as interference, providing a conservative bound. With perfect pilot cancellation, $\gamma_c=\mathcal A(P_s-p_d)/\mathcal B$ and $p_d^{\max}=P_s-\gamma_{\min}\mathcal B/\mathcal A$; only the feasible pilot-power limit changes, not the allocation rule. The optimization holds $P_s$, $P_r$, $g$, $h$, $\bar\rho_b$ and $\eta$ fixed, and hence also $\beta$.


\begin{proposition}[Kullback--Leibler-optimal pilot allocation]
Assume \eqref{eq:ratefeas} holds for a nondegenerate link with $g\ne0$, $a\ne0$, $P_r>0$, $\sigma_e^2>0$. Under the decorrelated residual-DPI model in \eqref{eq:raw}, the Kullback--Leibler (KL) divergence $D(\mathbb P_1\|\mathbb P_0)$ between the distributions $\mathbb P_i$ of $\{y[\ell]\}_{\ell=1}^L$ under $\HH_i$ increases strictly with $p_d$ for any residual DPI coefficient $\bar\rho_b\in\mathbb C$. Hence, $p_d=p_d^{\max}$ is KL-optimal.
\end{proposition}
\emph{Proof:} Let $G=|g|^2$, $K=|a\beta|^2$, $C=|\bar\rho_b\beta|^2$, $N_0=\sigma_e^2$, $\bar v=GP_s+\sigma_r^2$ and $v=\bar v-Gp_d$; then $\sigma_0^2=N_0+Cv$, $\sigma_1^2=N_0+(C+K)v$ and $E_\mu/L+\Delta=K\bar v$. This gives $D(\mathbb P_1\|\mathbb P_0)=L\{K\bar v/(N_0+Cv)-\ln[1+Kv/(N_0+Cv)]\}$, whose derivative is
\begin{equation}
 \frac{\partial D}{\partial p_d}=
 \frac{LGK\!\left[N_0^2+N_0C(\bar v+v)+C(C+K)\bar v v\right]}
 {(N_0+Cv)^2\left[N_0+(C+K)v\right]}>0.
\label{eq:klderivative}
\end{equation}
The strict positivity proves the claim. \hfill$\square$


\begin{proposition}[Matched-filter asymptote]
For $g\ne0$, $a\ne0$, $p_d>0$ and $\sigma_e^2>0$, the coherent matched-filter deflection
$d_{\rm MF}^2=2E_\mu/\sigma_1^2$ increases strictly with $P_r$ and approaches
\begin{equation}
 \lim_{P_r\to\infty}d_{\rm MF}^2
 =
 \frac{|a|^2}{|\bar\rho_b|^2+|a|^2}
 \frac{2Lp_d}{p_g+\sigma_r^2/|g|^2},
\label{eq:mfceiling}
\end{equation}
which reduces to $2Lp_d/(p_g+\sigma_r^2/|g|^2)$ under ideal DPI suppression, $\bar\rho_b=0$.
\end{proposition}
\emph{Proof:}
Substituting \eqref{eq:afmap} and
$\sigma_1^2=\sigma_e^2+\beta^2(|\bar\rho_b|^2+|a|^2)v$
shows that $2E_\mu/\sigma_1^2$ increases strictly with $\beta^2$, and hence with $P_r$. Dividing by $\beta^2$ and taking $P_r\to\infty$ gives \eqref{eq:mfceiling}.
\hfill$\square$


\begin{corollary}[Ideal-DPI limit of the full NP detector]
Under ideal DPI suppression ($\bar\rho_b=0$), with $\sigma_e^2>0$ and $|a|^2v>0$, the exact NP detector achieves $P_D\to1$ as $P_r\to\infty$.
\end{corollary}
\emph{Proof:} From \eqref{eq:pd}, $\sigma_0^2=\sigma_e^2$ remains fixed while $\Delta,E_\mu\propto\beta^2$. Thus, $\sigma_0^2/\sigma_1^2\to0$ and $\lambda_0=2\sigma_0^2E_\mu/\Delta^2\to0$. Since $\Qb^{-1}_{2L,\lambda_0}(\alpha)$ remains finite, the threshold argument $x\to0$, giving $P_D\ge\Qb_{2L,0}(x)\to1$. \hfill$\square$


\begin{proposition}[Residual-DPI ceiling]
For fixed $g$, $\bar\rho_b\ne0$, $p_d$, $p_g$, $\sigma_r^2$, with $a\ne0$ and $v>0$, the full NP detector approaches the following limit as $P_r\to\infty$:
\begin{equation}
 P_D^{\infty}=\Qb_{2L,\lambda_1^\infty}\!\left(
 \frac{c_0}{c_1}\Qb^{-1}_{2L,\lambda_0^\infty}(\alpha)\right),
\label{eq:npceiling}
\end{equation}
where $c_0=|\bar\rho_b|^2v$, $c_1=(|\bar\rho_b|^2+|a|^2)v$, $e=L|a|^2|g|^2p_d$, $c_\Delta=|a|^2v$, and $\lambda_i^\infty=2c_i e/c_\Delta^2$.
\end{proposition}
\emph{Proof:} Dividing $(\sigma_0^2,\sigma_1^2,E_\mu,\Delta)$ by $\beta^2$ in \eqref{eq:pd} and taking $P_r\to\infty$ gives the stated limit. \hfill$\square$

\subsection{Partially Correlated Residual Direct Path}
\label{sec:corr}

Model \eqref{eq:raw} assumes decorrelated residual DPI, whereas practical cancellation may retain part of the target-cell waveform. We model the residual as $\widetilde u=\eta u+\sqrt{1-\eta^2}\,u^{\perp}$, where $u$ contains the target-cell payload and relay noise, $u^{\perp}$ is an independent copy (proper complex Gaussian, variance $v$, white over $\ell$), and $\eta\in[0,1]$ is the post-processing correlation, not determined by the $0.08$-cell offset. The null variance $\sigma_0^2$ in \eqref{eq:moments} is unchanged, while
\begin{equation}
 \Delta(\eta)=\beta^2v\!\left(|a|^2+2\eta\,\Re\{a\bar\rho_b^*\}\right),
\label{eq:coherentdpi}
\end{equation}
which reduces to the decorrelated model at $\eta=0$ and the fully unresolved coherent model at $\eta=1$, where $\widetilde u=u$. 
Proposition~1 does not follow by substitution: with $\phi=\arg a-\arg\bar\rho_b$, correlation changes the variance increment to $K_\eta=K+2\eta\cos\phi\sqrt{KC}$ but leaves $E_\mu\propto|a|^2$, invalidating $E_\mu/L+\Delta=K\bar v$. Using $S=N_0+Cv$ and $S_\eta=N_0+(C+K_\eta)v$ gives
\begin{equation}
 \frac{\partial D}{\partial p_d}=LG\!\left[\frac{CK\bar v-(K_\eta-K)N_0}{S^2}
 +\frac{K_\eta N_0}{S\,S_\eta}\right],
\label{eq:klderivative2}
\end{equation}
which reduces to \eqref{eq:klderivative} when $K_\eta=K$. To establish positivity, let $\tilde K=Kv$, $\tilde C=Cv$, $\theta=\bar v/v\ge1$ ($v\ge\sigma_r^2>0$) and $m=\eta\cos\phi$. Multiplying \eqref{eq:klderivative2} by $S^2S_\eta v/(LG)$ gives $\tilde KJ$, where $J=-4N_0\tilde Cm^2+2\sqrt{\tilde K\tilde C}\,(\tilde C\theta-N_0)m+N_0^2+N_0\tilde C(1+\theta)+\tilde C^2\theta+\tilde K\tilde C\theta$ is concave in $m$. Minimizing over $m\in[-1,1]$ at an endpoint and then over $\sqrt{\tilde K}$ gives $J\ge N_0(\theta-1)(N_0/\theta+\tilde C)$, positive for $p_d>0$ and nonnegative at $p_d=0$. Hence, $D$ is strictly increasing on $[0,p_d^{\max}]$, making $p_d^{\max}$ KL-optimal at fixed $C$ for any correlation. This fixed-$C$ model follows from the independence of $\bar\rho_b$ and $\beta$ from $p_d$. In contrast, Figs.~\ref{fig:results}(f)--(g) fix the residual power $Cv$, hence $\sigma_0^2$, as $p_d$ varies; at the $+5$ dB level of Fig.~\ref{fig:results}(g), monotonicity holds only for $-0.81\le\eta\cos\phi\le0.75$. The completion in \eqref{eq:square} remains valid for $\Delta(\eta)\ne0$, but the relative phase may cause $\Delta(\eta)$ to cross zero, with a negative gap reversing the decision tail. For $\Delta(\eta)\ne0$, define $s=\operatorname{sgn}\Delta(\eta)$ and $\mathcal G_{s,\nu,\lambda}(x)=\Pr\{\chi'^2_\nu(\lambda)>x\}$ for $s=+1$ and $\Pr\{\chi'^2_\nu(\lambda)<x\}$ for $s=-1$,
\begin{equation}
 P_{D,\eta}=\mathcal G_{s,2L,\lambda_{1}}\!\left(
 \frac{\sigma_{0}^2}{\sigma_{1}^2}
 \mathcal G^{-1}_{s,2L,\lambda_{0}}(\alpha)\right),~~
 \lambda_{i}=\frac{2\sigma_{i}^2E_\mu}{\Delta(\eta)^2}.
\label{eq:coherentpd}
\end{equation}
At $\Delta(\eta)=0$, the NP detector reduces to matched filtering. This occurs at $\eta_0=-|a|^2/(2\Re\{a\bar\rho_b^*\})=|a|/(2|\bar\rho_b||\cos\phi|)$ for $\cos\phi<0$ and $\eta_0\le1$; $|a|/(2|\bar\rho_b|)$ applies only at $\phi=\pi$. Since the pilot mean remains, $P_D$ reaches the matched-filter value $Q(Q^{-1}(\alpha)-\sqrt{2E_\mu/\sigma_0^2})$, not $P_{\rm FA}$: $0.0154$, or $15P_{\rm FA}$, in Fig.~\ref{fig:results}(g). Beyond $\eta_0$, detection recovers on the reversed tail. Thus, residual-DPI correlation determines both the variance-gap sign and the optimal decision tail.

\begin{figure}[!t]
\centering
\includegraphics[width=\linewidth,keepaspectratio]{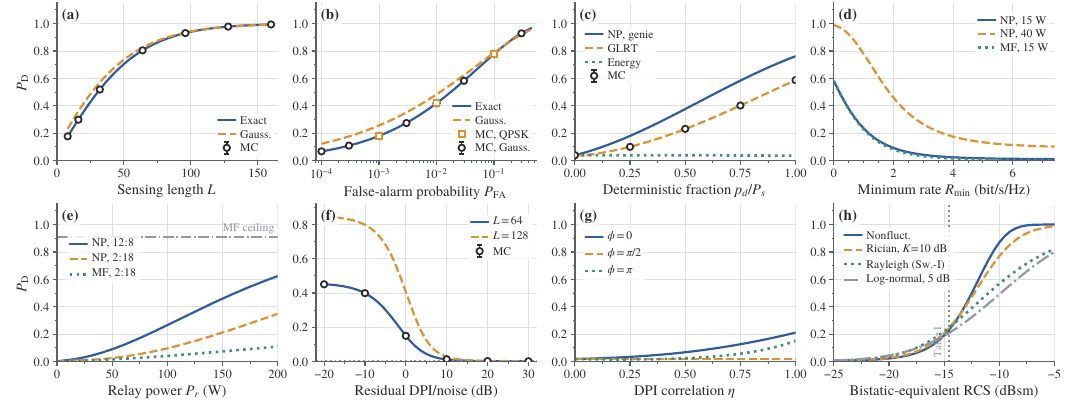}
\vspace{-5pt}
\caption{Link-budget analysis and raw AF-cascade simulation; Monte Carlo (MC) markers carry $95\%$ Wilson intervals and every threshold is analytic. (a) Finite-$L$ validation at $P_{\rm FA}=0.02$. (b) Low-$P_{\rm FA}$ accuracy at $L=32$, with QPSK markers. (c) Unknown echo phase: GLRT \eqref{eq:glrt}, genie bound \eqref{eq:pd} and energy detector. (d) Rate tradeoff at $L=64$. (e) Relay power at $L=8$ for $p_d{:}p_g$ of $12{:}8$ and $2{:}18$ W; MF ceiling \eqref{eq:mfceiling}. (f) Decorrelated residual DPI. (g) Partially correlated residual DPI, \eqref{eq:coherentpd}, at $+5$ dB; at $\phi=\pi$ the gap crosses zero and $P_D$ floors at the matched-filter value. (h) Fluctuating target over a $\pm10$ dB RCS sensitivity range, GLRT at $p_d/P_s=0.5$ matched to the mean $|a|^2$ on the abscissa. Defaults: $P_r=40$ W, $P_{\rm FA}=10^{-3}$, $L=32$, ideal DPI suppression, $p_d=p_d^{\max}$ at $R_{\min}=2$ bit/s/Hz.}
\label{fig:results}
\end{figure}


\section{Link Budget and Results}
We use the reproducible free-space link budget in Table~\ref{tab:budget} instead of normalized channel gains; the 600-km satellite--relay range is consistent with the LEO class in \cite{3gpp38181}. We choose $\sigma_b=0.035$ m$^2$ ($-14.56$ dBsm) as an illustrative bistatic-equivalent RCS: \cite{ezuma2020rcs} reports $-14.24$ to $-15.03$ dBsm for three small drones at $15$ GHz, based on monostatic measurements with strong aspect dependence. Fig.~\ref{fig:results}(h) therefore varies $\sigma_b$ over $\pm10$ dB around this value. For a one-way link of length $d$ with antenna gains $G_t$, $G_r$ and loss $L_x$, and for the bistatic echo,
\begin{align}
 |q(d)|^2&=10^{(G_t+G_r-L_x)/10}\left(\frac{\lambda_c}{4\pi d}\right)^2,
\label{eq:friis}\\
 |a|^2&=\frac{10^{(G_{r,t}+G_{s,r}-L_e)/10}\lambda_c^2\sigma_b}
 {(4\pi)^3d_{rt}^2d_{ts}^2},\quad \sigma_n^2=kT_0B10^{F/10}.
\label{eq:radarbudget}
\end{align}
Here $\lambda_c$ is the wavelength; $G_{r,t}$ and $G_{s,r}$ are the relay-transmit and sensing-receive gains; $L_e$ is the echo loss; $d_{rt}$ and $d_{ts}$ are the two bistatic path lengths; $k$ is Boltzmann's constant; and $F$ is the noise figure of the relay, user or sensing receiver for $\sigma_r^2$, $\sigma_u^2$ or $\sigma_e^2$. Applying \eqref{eq:friis} to $d_{sr}$, $d_{ru}$ and $d_{rs}$ gives $|g|^2$, $|h|^2$ and $|b|^2$, respectively.

\begin{table}[t]
\centering
\caption{Reproducible 15-GHz SAGIN Setting; Fig.~\ref{fig:results}(e) sweeps $P_r$ to 200 W.}
\label{tab:budget}
\footnotesize
\setlength{\tabcolsep}{3.0pt}
\begin{tabular}{l l}
\hline
Parameter & Value \\
\hline
Carrier, bandwidth $B$, temperature $T_0$ & 15 GHz, 100 kHz, 290 K \\
Satellite power; relay power & 20 W; 40 W (nominal) \\
Satellite-transmit / relay-receive gains & 38 / 30 dBi \\
Relay-transmit / user / sensing-receive gains & 15 / 3 / 20 dBi \\
$d_{sr},d_{ru},d_{rt}\!=\!d_{ts},d_{rs}$ & 600 km, 2, 1.118, 2 km \\
Target offset from baseline & 500 m ($26.6^\circ$ split) \\
Relay / user / sensing noise figures & 5 / 5 / 4 dB \\
Link losses: $sr,ru,$ echo, direct & 2 / 2 / 4 / 2 dB \\
Illustrative bistatic-equivalent UAV RCS $\sigma_b$ & 0.035 m$^2$ ($-14.56$ dBsm) \\
\hline
\end{tabular}
\end{table}

This budget gives $P_s|g|^2=-62.52$ dBm at the relay, $|a|^2P_r=-126.44$ dBm at $P_r=40$ W, and a sensing-noise power of $-119.98$ dBm. The uncancelled relay--sensor direct-path signal is $77.01$ dB above the sensing-noise floor, already accounting for the $121.99$ dB free-space loss over $d_{rs}$ and the $15$ and $20$ dBi antenna gains. Thus, a residual-DPI/noise level of $x$ dB in Fig.~\ref{fig:results}(f) requires $77.01-x$ dB of \emph{additional} isolation and cancellation, or $47$--$77$ dB for $x\in[0,30]$ dB. The echo and direct path arrive $26.6^\circ$ apart against a $3$-dB beamwidth near $17^\circ$; angular discrimination may help but is not credited. Analog suppression must precede digital cancellation to avoid receiver saturation, which is not modeled. The in-band full-duplex relay is assumed to suppress its own self-interference before \eqref{eq:raw}. We assume low clutter: Fig.~\ref{fig:results}(h) examines target fluctuation, not ground multipath or urban clutter, which cannot generally be modeled by increasing $\sigma_e^2$.



Monte Carlo markers use raw samples from \eqref{eq:raw}, not the derived chi-square law: $4.5\times10^4$ trials for Fig.~\ref{fig:results}(a), $6\times10^4$ for (c) and (f), and $1.8\times10^5$ and $4\times10^5$ for Gaussian and QPSK markers in (b). Analytic thresholds are used throughout; the widest $95\%$ Wilson interval is $\pm0.0046$. Separate $10^6$-trial checks at the setting of Fig.~\ref{fig:results}(b) confirm detection and false-alarm performance, with $P_{\rm FA}\in[0.97,1.02]\times10^{-3}$ for Gaussian, QPSK, and 16-ary quadrature amplitude modulation (16-QAM).


Figs.~\ref{fig:results}(a)--(b) confirm \eqref{eq:pd} within the $95\%$ Wilson intervals; at $P_r=40$ W, $L=32$ and $P_{\rm FA}=10^{-4}$, the Gaussian approximation overestimates $P_D$ by $0.056$. Fig.~\ref{fig:results}(c) shows that the NP detector combines payload-echo energy and coherent pilot correlation. Fig.~\ref{fig:results}(d) illustrates \eqref{eq:ratefeas} and the KL-optimal allocation \eqref{eq:pdmax}--\eqref{eq:klderivative}: a higher minimum rate reduces $P_D$, although the payload echo remains useful. Fig.~\ref{fig:results}(e) shows that the MF detector remains below its asymptotic limit over the simulated power range. At $200$ W, its deflection is $0.90$ and $\sigma_0/\sigma_1=0.70$, compared with the limiting values of $1.78$ and $0$. Thus, \eqref{eq:mf} gives $P_D^{\rm MF}=0.11$, versus the limiting $0.91$. Under ideal DPI suppression, the full NP detector has no such ceiling. The extended power sweep illustrates these trends rather than a practical $200$ W relay design. Fig.~\ref{fig:results}(f) shows the effect of residual DPI: the leakage level at which $P_D$ falls to $0.2$ increases from $-1.6$ dB at $L=64$ to $+3.3$ dB at $L=128$, indicating about $5$ dB more leakage tolerance when the dwell length is doubled.

\subsection{Receiver Implementation}
No online optimization is needed: for fixed $(L,\alpha,\sigma_0^2,\sigma_1^2,E_\mu)$, the threshold $u_\alpha=\tfrac{1}{2}\sigma_0^2\Qb^{-1}_{2L,\lambda_0}(\alpha)$ can be tabulated offline, while evaluating \eqref{eq:T} requires one energy accumulation and one pilot correlation, giving $O(L)$ complexity with constant memory. Two practical issues remain. First, $\sigma_0^2$ must be estimated. With $\widehat\sigma_0^2=N^{-1}\sum_{n=1}^{N}|y_n|^2$ computed from $N$ independent target-free reference samples, one per cell, $\xi=\widehat\sigma_0^2/\sigma_0^2$ satisfies $2N\xi\sim\chi^2_{2N}$. The estimate affects both the statistic and its threshold: with $\Delta$, $E_\mu$ fixed, the plug-in statistic is $\widehat U=\sum_\ell|y[\ell]+\xi\kappa\mu_\ell|^2$, whose \emph{actual} null law, normalized by the true $\sigma_0^2$, is $2\widehat U/\sigma_0^2\,|\,\xi\sim\chi'^2_{2L}(\xi^2\lambda_0)$. The receiver, however, calibrates under its \emph{assumed} model, in which $2\widehat U/\widehat\sigma_0^2\sim\chi'^2_{2L}(\xi\lambda_0)$. Thus,
\begin{equation}
 \alpha_N=\E_\xi\!\left[\Qb_{2L,\xi^2\lambda_0}\!\left(\xi\,\Qb^{-1}_{2L,\xi\lambda_0}(\alpha)\right)\right].
\label{eq:cfar}
\end{equation}
The two effects partly offset, but false-alarm inflation remains. Table~\ref{tab:cfar} shows $\alpha_N/\alpha=102$ at $N=16$, $\alpha=10^{-4}$, falling below $3$ only for $N\ge128$. The factor $c_N$ scales the plug-in threshold $\tfrac{1}{2}\widehat\sigma_0^2\Qb^{-1}_{2L,\xi\lambda_0}(\alpha)$ to restore $\alpha_N=\alpha$, at some cost in $P_D$. This applies to the known-phase plug-in detector, not the GLRT in Section~\ref{sec:mismatch}. Second, scanning $N_c$ approximately independent cells requires $\alpha=1-(1-\alpha_{\rm scan})^{1/N_c}$; correlated cells require an effective-cell correction.

\begin{table}[t]
\centering
\caption{Realized $P_{\rm FA}$ With $N$ i.i.d.\ Reference Samples ($L=32$, $\lambda_0=126$)}
\label{tab:cfar}
\footnotesize
\setlength{\tabcolsep}{3.4pt}
\begin{tabular}{l c c c c c}
\hline
$N$ & 16 & 32 & 64 & 128 & 512 \\
\hline
$\alpha_N/\alpha$, $\alpha=10^{-3}$ & $22.7$ & $7.5$ & $3.2$ & $1.9$ & $1.2$ \\
$\alpha_N/\alpha$, $\alpha=10^{-4}$ & $102$ & $21.7$ & $6.0$ & $2.6$ & $1.3$ \\
$c_N$, $\alpha=10^{-4}$ & $1.53$ & $1.18$ & $1.07$ & $1.03$ & $1.01$ \\
\hline
\end{tabular}
\end{table}


\section{Target-Coefficient Uncertainty}
\label{sec:mismatch}
\subsection{Unknown Echo Phase}
With $B=100$ kHz, the range resolution is $c/2B=1499$ m, corresponding to $7.5\times10^4$ wavelengths at $15$ GHz, so the range cell alone does not determine $\arg a$. We model the echo phase as an unknown deterministic parameter without prior phase information and maximize the likelihood over it; phase estimates from previous observations could also be used. The magnitude is better constrained: \eqref{eq:radarbudget} gives $|a|^2$ for an assumed RCS, although the actual RCS may vary. Since $\Delta$, $w$ and $\kappa$ depend only on $|a|$, maximizing the log-likelihood ratio over the phase affects only the coherent term in \eqref{eq:T}, replacing $\Re\{\cdot\}$ by its modulus. Defining $W=L^{-1/2}\sum_\ell c^*[\ell]y[\ell]$, $R=|W|$ and $V=\sum_\ell|y[\ell]|^2-R^2$, projection onto $\mathbf c/\sqrt L$, with $\mathbf c=[c[1],\dots,c[L]]^T$, makes $V$ independent of $W$. Completing the square as in \eqref{eq:square} then gives $U_G=(R+\omega)^2+V$, $\omega=\kappa\sqrt{E_\mu}$.


\looseness=-1
This derivation requires $\Delta$, $w$, $\kappa$ to depend only on $|a|$, which is not the case in \eqref{eq:coherentdpi} when $\eta\bar\rho_b\ne0$. In that case, $\Delta(\eta)$ depends on $\arg a$, so maximizing over the phase also changes the variance and normalization. Thus, \eqref{eq:coherentpd} remains a known-coefficient benchmark.


\begin{proposition}[Exact unknown-phase GLRT]
For $L\ge2$, $\eta\bar\rho_b=0$, unknown echo phase and assumed $|a|$, under $\HH_i$, $R$ follows a Rice distribution with noncentrality $\delta_i\sqrt{E_\mu}$, scale $\sigma_i/\sqrt2$ and density $f_{R_i}$, while $V\sim(\sigma_i^2/2)\chi^2_{2L-2}$ is independent of $R$. Thus, for $u>\omega^2$ and $r_\star=\sqrt u-\omega$,
\begin{align}
 \Pr\{U_G>u\}=&\int_0^{r_\star}\!\Qb_{2L-2,0}\!\left(\frac{2\left(u-(r+\omega)^2\right)}{\sigma_i^2}\right)\!f_{R_i}(r)\,dr\nonumber\\
 &+\Qb_{2,\lambda_i^W}\!\left(2r_\star^2/\sigma_i^2\right),\quad \lambda_i^W=\frac{2\delta_i^2E_\mu}{\sigma_i^2},
\label{eq:glrt}
\end{align}
and $\Pr\{U_G>u\}=1$ for $u\le\omega^2$. Setting the threshold using \eqref{eq:glrt} under $\HH_0$ to achieve $\alpha$ and evaluating the same expression under $\HH_1$ gives the exact $P_D$ for a Gaussian payload, given $|a|$ and the true $\sigma_0^2$.
\end{proposition}


Both tail probabilities involve single integrals, so the finite-$L$ exactness of \eqref{eq:pd} is retained, with \eqref{eq:pd} serving as the genie bound. The performance loss arises entirely from the coherent component and vanishes without a pilot: Fig.~\ref{fig:results}(c) shows that the two detectors coincide with the energy detector at $p_d=0$ but diverge as the pilot power increases.


\subsection{Plug-In Mismatch and Target Fluctuation}
Consider a detector using $\widehat a=\varrho e^{j\psi}a$, giving $\widehat\mu_\ell=\varrho e^{j\psi}\mu_\ell$, $\widehat\Delta=\varrho^2\Delta>0$ and $\widehat\kappa=\sigma_0^2/\widehat\Delta$, while the actual target-present distribution remains $\CN(\mu_\ell,\sigma_1^2)$. Calibrating the false-alarm probability exactly for $U^{\rm mis}=\sum_\ell|y[\ell]+\widehat\kappa\widehat\mu_\ell|^2$ gives
\begin{equation}
 P_D^{\rm mis}=\Qb_{2L,\zeta_1}\!\left(
 \frac{\sigma_0^2}{\sigma_1^2}\Qb^{-1}_{2L,\zeta_0}(\alpha)\right),
\label{eq:mismatch}
\end{equation}
where $\zeta_0=2\widehat\kappa^2\varrho^2E_\mu/\sigma_0^2$ and $\zeta_1=2E_\mu(1+\widehat\kappa^2\varrho^2+2\widehat\kappa\varrho\cos\psi)/\sigma_1^2$. Mismatch changes only the noncentrality parameters of the chi-square distributions, and $\varrho=1$, $\psi=0$ recovers \eqref{eq:pd}. For the known-phase detector, Fig.~\ref{fig:mismatch} gives a matched detection probability of $0.3758$: a $\pm6$ dB magnitude error reduces it by at most $0.016$, whereas a $60^\circ$ phase error alone causes a reduction of $0.30$. A magnitude overestimate can even improve detection at large $\psi$: increasing $\widehat\Delta$ reduces $\widehat\kappa$ and gives less weight to the mismatched coherent term. The GLRT is insensitive to phase error, and over the same $\pm6$ dB magnitude range its $P_D$ varies by at most $0.010$ around $0.230$, as shown by the flat band in Fig.~\ref{fig:mismatch}(b). These results support estimating the phase while using an assumed magnitude, rather than treating the magnitude as perfectly known.

For target fluctuations with $|a|^2$ fixed within each dwell, $P_D$ is averaged over $|a|^2$, while the false-alarm rate under $\HH_0$ remains unchanged. Fig.~\ref{fig:results}(h) compares non-fluctuating, Rician ($10$-dB factor), Rayleigh (Swerling-I) \cite{swerling1960,song2013rcs}, and log-normal models fitted to $15$-GHz drone measurements \cite{ezuma2020rcs}. All use the same linear mean $|a|^2$; the log-normal standard deviation is $5$ dB for $10\log_{10}|a|^2$. Detector weights and thresholds are set at the nominal mean, not adjusted for each realization.

\begin{figure}[t]
\centering
\includegraphics[width=0.90\linewidth,keepaspectratio]{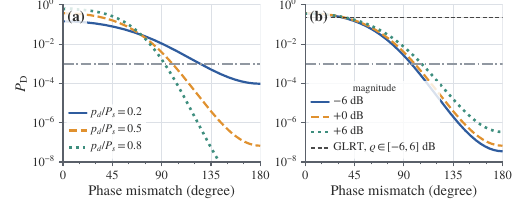}
\vspace{-5pt}
\caption{Exact plug-in mismatch of the \emph{known-phase} detector, \eqref{eq:mismatch}, at $P_r=40$ W, $L=32$, $P_{\rm FA}=10^{-3}$, ideal DPI suppression. (a) Phase error; pilot-heavy operation is more phase-sensitive. (b) Joint error at $p_d/P_s=0.5$; dashed band: unknown-phase GLRT over the same $\pm6$ dB magnitude error, flat in $\psi$.}
\label{fig:mismatch}
\end{figure}

\section{Conclusion}
This letter investigated bistatic UAV detection in SAGINs, where an AF relay supports communication while also illuminating a target. We derived exact finite-sample detection probabilities by accounting for the target-reflected pilot, payload, and relay noise. The analysis showed that reflected relay noise provides useful variance information even without payload transmission, while relay gain normalization couples communication and sensing performance. We further characterized the effects of relay power and residual direct-path interference, and developed an exact phase-only GLRT for an assumed target magnitude. The results provide a basis for practical dual-use AF sensing in SAGINs, while joint magnitude--phase estimation and clutter remain open challenges.



\begin{thebibliography}{6}

\bibitem{liu2022isac}
F. Liu, Y. Cui, C. Masouros, J. Xu, T. X. Han, Y. C. Eldar,
and S. Buzzi,
``Integrated sensing and communications: Toward dual-functional
wireless networks for 6G and beyond,''
\emph{IEEE J. Sel. Areas Commun.},
vol. 40, no. 6, pp. 1728--1767, Jun. 2022.

\bibitem{xiong2023tradeoff}
Y. Xiong, F. Liu, Y. Cui, W. Yuan, T. X. Han, and G. Caire,
``On the fundamental tradeoff of integrated sensing and
communications under Gaussian channels,''
\emph{IEEE Trans. Inf. Theory},
vol. 69, no. 9, pp. 5723--5751, Sep. 2023.

\bibitem{cheng2025networked}
G. Cheng, X. Song, Z. Lyu, and J. Xu,
``Networked ISAC for low-altitude economy: Coordinated transmit
beamforming and UAV trajectory design,''
\emph{IEEE Trans. Commun.},
vol. 73, no. 8, pp. 5832--5847, Aug. 2025.

\bibitem{song2026lae}
X. Song, X. Yu, J. Xu, and D. W. K. Ng,
``Integrated sensing and communications for low-altitude economy
with deterministic sensing and Gaussian information signals,''
arXiv:2604.19040, Apr. 2026.

\bibitem{xie2025bistatic}
L. Xie, F. Liu, S. Song, and S. Jin,
``Bistatic target detection by exploiting both deterministic pilots
and unknown random data payloads,''
arXiv:2508.18728, Aug. 2025.

\bibitem{song2025activeirs}
X. Song, X. Qin, X. Yu, J. Xu, and D. W. K. Ng,
``Active-IRS-enabled target detection,''
\emph{IEEE Wireless Commun. Lett.},
vol. 14, no. 9, pp. 2763--2767, Sep. 2025.

\bibitem{bhola2026amud}
Bhola, Y.-J. Chen, A. Balakrishnan, S. De, and L.-C. Wang,
``Adaptive multi-UAV relay deployment framework in satellite
aerial ground integrated systems,''
\emph{IEEE Trans. Veh. Technol.}, early access, 2026.

\bibitem{3gpp38181}
3GPP,
``NR; Satellite access node conformance testing,''
3GPP TS 38.181, V17.1.0, Release 17, Jul. 2023.

\bibitem{ezuma2020rcs}
M. Ezuma, M. Funderburk, and I. Guvenc,
``Compact-range RCS measurements and modeling of small drones
at 15 GHz and 25 GHz,''
in \emph{Proc. IEEE Radio Wireless Symp. (RWS)},
San Antonio, TX, USA, Jan. 2020, pp. 313--316.

\bibitem{swerling1960}
P. Swerling,
``Probability of detection for fluctuating targets,''
\emph{IEEE Trans. Inf. Theory},
vol. 6, no. 2, pp. 269--308, Apr. 1960.

\bibitem{song2013rcs}
X. Song, W. D. Blair, P. Willett, and S. Zhou,
``Dominant-plus-Rayleigh models for RCS: Swerling III/IV versus Rician,''
\emph{IEEE Trans. Aerosp. Electron. Syst.},
vol. 49, no. 3, pp. 2058--2064, Jul. 2013.

\end{thebibliography}
\end{document}